\documentclass[USenglish]{lipics}

\pdfoutput=1

\usepackage[T1]{fontenc}		% output font encoding
\usepackage[utf8]{inputenc}		% input encoding

\usepackage{tikz}
	\usetikzlibrary{calc,arrows,matrix,positioning,fit,decorations.markings}

\usepackage{enumitem}		% paramétrage itemize
	\setitemize{itemsep=0pt}
	\setenumerate{itemsep=0pt}
	\newlength{\savedparindent}
	\setitemize{leftmargin=\savedparindent}
	\setenumerate{leftmargin=\savedparindent}
	\setitemize{topsep=\parskip}
\usepackage[capitalize]{cleveref}		% automatic references
\usepackage{aliascnt}

\usepackage{xspace}			% gestion des espaces après les macro de texte
\usepackage{microtype}		% for typography geeks
\usepackage[autostyle = true, english = american, strict = true, autopunct=true]{csquotes}		% guillemets
\usepackage{relsize}		% relative sizes in math 

\usepackage{amssymb}			% AMS stuff
\usepackage{mathtools}			% other math commands (loads AMSmath)
\usepackage{ebproof}		% EB package for proofs, replaces bussproofs
	\newcommand{\seqrule}{\footnotesize\texttt}	% font for rules label
	\newcommand{\Theproof}[2]{\Hypo{\langle#1\rangle}\Infer[rulestyle=none]{1}{#2}}
	\setkeys{ebproof}{labelsep=1pt}
	\setkeys{ebproof}{rightlabeltemplate=\footnotesize #1}
	\hypersetup{bookmarksopen=true,bookmarksdepth=2}
\usepackage{cmll}			% EB package for LL
\usepackage{stmaryrd}		% for the parr and with
\usepackage{pn}

\AtBeginDocument{
	\hyphenpenalty=1000
	\tolerance=400
	\setlength{\parskip}{2pt}
}

\undef{\lemma}
\undef{\corollary}
\undef{\definition}
\undef{\example}
\undef{\remark}

\theoremstyle{plain}
\newtheorem{lemma}[theorem]{Lemma}
\newtheorem{corollary}[theorem]{Corollary}
\theoremstyle{definition}
\newtheorem{definition}[theorem]{Definition}
\newtheorem{example}[theorem]{Example}
\newtheorem{notation}[theorem]{Notation}
\theoremstyle{remark}
\newtheorem{remark}[theorem]{Remark}

\newenvironment{acknowledgement}
{\vskip\bigskipamount\noindent
\textcolor{darkgray}{\fontsize{9.5}{12.5}\sffamily\bfseries Acknowledgments\enskip}\fontsize{9.5}{12.5}}
{}

\renewcommand{\neg}{\overline}

\newcommand{\labwith}[1]{\with\!_{#1}}
\newcommand{\rlabwith}[1]{\rwith_{#1}}
\newcommand{\tlabwith}[1]{\twith_{#1}}
\newcommand{\protectobdd}{{\sf\textsf{\textsuperscript oBDD}}\xspace}
\newcommand{\obdd}[1]{%
		\ifthenelse{\equal{}{#1}}%
		{{\sf\textsf{\textsuperscript oBDD}}}%
		{{\sf\textsf{\textsuperscript oBDD}\!{\footnotesize $/#1$}}}%
	}
\newcommand{\bindd}{{\sf \textsf{BDD}}\xspace}
\newcommand{\aco}{\compclass{AC}$_0$\xspace}
\newcommand{\one}{\unit}
\newcommand{\ineg}[1]{\widehat{#1}}
\newcommand{\idnf}[1]{#1_{\tt m}}
\newcommand{\itef}[3]{{\tt If}\:#1\:{\tt Then}\:#2\:{\tt Else}\:#3}
\newcommand{\dcare}[2]{{\tt DontCare}\:#1\:{\tt Then}\:#2}

\newcommand{\ite}{\bindd}

\newcommand{\bdd}{\BDD}
\newcommand{\BDD}{\bindd}
\newcommand{\mallm}{{\bf\textsf{MALL}}\textsuperscript{\!\bf-}\xspace}
\newcommand{\mllm}{{\bf\textsf{MLL}}\textsuperscript{\!\bf-}\xspace}
\newcommand{\mll}{{\bf\textsf{MLL}}\xspace}
\newcommand{\lalpha}{\alpha^{\tt l}}
\newcommand{\ralpha}{\alpha^{\tt r}}
\newcommand{\zero}{\mathbf 0}
\newcommand{\Hughes}{D.~Hughes\xspace}
\newcommand{\VGlab}{R.~van~Glabbeek\xspace}

\newcommand{\twith}{\rwith}
\newcommand{\slfont}{\mathcal}
\newcommand{\bslicing}[1]{\slfont B_{#1}}
\newcommand{\mslicing}[1]{\slfont B_{#1}}
\newcommand{\slicing}[1]{\slfont S_{#1}}
\newcommand{\malleq}{\sim}
\newcommand{\booleq}{\sim}
\newcommand{\sliceeq}{\sim}
\newcommand{\nalpha}{\lneg\alpha}
\newcommand{\nbeta}{\lneg\beta}

\newcommand{\ngamma}{\lneg\gamma}
\newcommand{\ndelta}{\lneg\delta}

\newcommand{\lbeta}{\beta^{\tt l}}
\newcommand{\rbeta}{\beta^{\tt r}}
\newcommand{\seqv}[1]{\Gamma^{#1}}

\newcommand{\seqvd}[1]{\Delta^{#1}}
\newcommand{\monof}[1]{\mathbf{#1}}

\newcommand{\pdnf}[1]{\pi_{#1}}
\newcommand{\pwith}[1]{\pi_{\rwith}^{#1}}

\newcommand{\pleft}{\pi_\mathbf{0}}
\newcommand{\pright}{\pi_\mathbf{1}}

\newcommand{\modulo}{\emph{modulo}\xspace}
\newcommand{\cut}{\texttt{cut}\xspace}

\newcommand{\locit}[1]{#1}

\newcommand{\via}{\locit{via}~}

\newcommand{\wloss}{\locit{w.l.o.g.}~}

\newcommand{\etc}{\locit{etc.}\xspace}

\newcommand{\compclass}[1]{{\sf\textbf{#1}}}		% complexity class
\newcommand{\problem}[1]{{\bf{#1}}}		% problem: GraphCylce etc.

\newcommand{\eqdef}{=}							% definition
\newcommand{\extlist}[2]{#1,\,\dots\,,#2}					% list/sequence in extension
\newcommandx{\extset}[3][1=]{#1\{#2,\,\dots\,,#3#1\}}		% set in extension
\newcommandx{\set}[3][1=]{#1\{\:#2 \ #1|\ #3\:#1\}}		% set in comprehension
\newcommand{\void}{\varnothing}									% emptyset
\newcommand{\boolf}{\mathbf B}								% boolean
\newcommand{\unit}{\mathbf 1}								% unit
\newcommand{\mall}{{\bf\textsf{MALL}}\xspace}

\newcommand{\lneg}[1]{{#1}^\star}				% linear negation
\newcommand{\rwith}{\binampersand}
\newcommand{\rparr}{\bindnasrepma}
\renewcommand{\with}{\,\binampersand\,}
\renewcommand{\parr}{\,\bindnasrepma\,}
\newcommand{\tensor}{\otimes}
\newcommand{\plus}{\oplus}
\newcommand{\lplus}{\oplus_{\tt l}}
\newcommand{\rplus}{\oplus_{\tt r}}

\newcommand{\Logspace}{\compclass{Logspace}\xspace}						% Logspace
\newcommand{\Pspace}{\compclass{Pspace}\xspace}
\newcommand{\coNP}{\compclass{coNP}\xspace}
\newcommand{\NP}{\compclass{NP}\xspace}

\newcommand{\p}{\textbf{\Large.}}

\begin{document}

\title{\mall proof equivalence is \Logspace-complete, \via binary decision diagrams}

\author{Marc Bagnol}

\affil{Department of Mathematics and Statistics -- University of Ottawa}

\authorrunning{M.\,Bagnol}
\Copyright{Marc Bagnol}



\maketitle

\begin{abstract}
	Proof equivalence in a logic is the problem of deciding whether two proofs are equivalent \modulo
a set of permutation of rules that reflects the commutative conversions of its cut-elimination procedure.
As such, it is related to the question of proofnets: finding canonical representatives of equivalence
classes of proofs that have good computational properties.
It can also be seen as the word problem for the notion of
free category corresponding to the logic.

It has been recently shown that proof equivalence in \mll (the multiplicative with units fragment of linear logic)
is \Pspace-complete,
which rules out any low-complexity notion of proofnet for this particular logic.

Since it is another fragment of linear logic for which attempts to define a fully satisfactory
low-complexity notion of proofnet have
not been successful so far,
we study proof equivalence in \mallm (multiplicative-additive without units fragment of linear logic)
and discover a situation that is totally different from the \mll case.
Indeed, we show that
proof equivalence in \mallm corresponds (under \aco reductions) to equivalence of binary decision diagrams,
a data structure widely used to represent and analyze Boolean functions efficiently.

We show these two equivalent problems to be \Logspace-complete. If this technically
leaves open the possibility for a complete solution to the question of proofnets for \mallm, 
the established relation with binary decision diagrams actually suggests a negative solution to this problem.

\end{abstract}

\begin{acknowledgement}
	to people from 
	\href{http://cstheory.stackexchange.com/questions/29243/what-is-the-complexity-of-the-equivalence-problem-for-read-once-decision-trees/29248#29248}{{\tt cstheory.stackexchange.com}}
	for pointing the author to the notion of \bdd;
	to \href{http://boole.stanford.edu/~dominic}{Dominic Hughes} for the live feedback during the redaction of
	the article;
	to \href{http://aubert.perso.math.cnrs.fr/}{Clément Aubert}, for his help in understanding \aco reductions.
\end{acknowledgement}

\section{Introduction}
	\subsection*{Proofnets: from commutative conversions to canonicity}\label{sec_twoface}
	%In proof theory, the want for a fine-grained understanding of the dynamics at work in the \cut-elimination
%procedure led to the introduction of linear logic~\cite{Girard1987}. In this logic, the structural rules
%(contraction, the rule that allows to copy information, in particular) are only available in certain situations,
%namely on formulas of the form $\oc A$, the modality $\oc$ being called the \emph{exponential} modality of 
%linear logic.

%The last decades have witnessed a shift in the focus of proof theory towards a more operational point of view,
%triggered by the \cut-elimination theorem by Gentzen.
%This results and its interpretation computational interpretation allow to view
%the \cut-elimination theorem can be regarded as a computation system.

From the perspective of the Curry-Howard (or propositions-as-types) correspondence~\cite{Gallier1995}, 
a proof of $A\Rightarrow B$ in a logic enjoying a 
\cut-elimination procedure can be seen as a program that inputs (through the \cut rule) a proof of $A$ 
and outputs a \cut-free proof of $B$. 

Coming from this dynamic point of view, linear logic~\cite{Girard1987}
makes apparent the distinction between data that can or cannot be copied/erased \via 
its exponential modalities and retains the symmetry of classical logic:
the linear negation $\lneg{(\cdot)}$ is an involutive operation.
The study of \cut-elimination is easier in this setting thanks to the linearity constraint. However,
in its sequent calculus presentation, the \cut-elimination procedure of linear logic still
suffers from the common flaw of these type of calculi: commutative conversions.

\begin{center}
\begin{prooftree}
	\Theproof{\pi}{\lneg A , \lneg B, C,D,\Gamma}
	\Infer{1}[$\parr$]{\lneg A\parr \lneg B,C,D,\Gamma}
	\Infer{1}[$\parr$]{\lneg A\parr \lneg B,C\parr D,\Gamma}
	\Theproof{\mu}{A}
	\Theproof{\nu}{B}
	\Infer{2}[$\otimes$]{A\otimes B}
	\Infer{2}[\tt cut]{C\parr D,\Gamma}
\end{prooftree}
\quad$\rightarrow$\qquad
\begin{prooftree}
	\Theproof{\pi}{\lneg A , \lneg B, C,D,\Gamma}
	\Infer{1}[$\parr$]{\lneg A\parr \lneg B,C, D,\Gamma}
	\Theproof{\mu}{A}
	\Theproof{\nu}{B}
	\Infer{2}[$\otimes$]{A\otimes B}
	\Infer{2}[\tt cut]{C,D,\Gamma}
	\Infer{1}[$\parr$]{C\parr D,\Gamma}
\end{prooftree}
\end{center}

In the above reduction, one of the two formulas related by the \cut rule is introduced deeper in the proof,
making it impossible to perform an actual elimination step right away: one needs first to \emph{permute}
the rules in order to be able to go on.

This type of step is called a \emph{commutative conversion} and their presence complexify a lot the study of the 
\cut-elimination procedure,
as one needs to work \modulo an equivalence relation on proofs that is not {orientable} into a rewriting 
procedure in an obvious way: there are for instance situations of the form
\begin{center}
{
\begin{prooftree}
	\Theproof{\pi_1}{\lneg A,\lneg B, \Gamma}
	\Theproof{\pi_2}{A}
	\Infer{2}[\tt cut]{\vdash\lneg B,\Gamma}
	\Theproof{\pi_3}{B}
	\Infer{2}[\seqrule{cut}]{\vdash\Gamma}
\end{prooftree}
}
\quad~$\leftrightarrow$~\qquad
%\scalebox{0.83}
{
\begin{prooftree}
	\Theproof{\pi_1}{\lneg A,\lneg B, \Gamma}
	\Theproof{\pi_3}{B}
	\Infer{2}[\tt cut]{\vdash\lneg A,\Gamma}
	\Theproof{\pi_2}{A}
	\Infer{2}[\tt cut]{\vdash\Gamma}
	
\end{prooftree}
}
\end{center}
where it is not possible to favor one side of the equivalence without further non-local knowledge of the proof.
The point here is that, as a language for describing proofs, sequent calculus is somewhat \emph{too explicit}.
For instance, the fact that the two proofs
\begin{center}
\begin{prooftree}
	\Theproof{\pi}{A , B, C,D, \Gamma}
	\Infer{1}[$\parr$]{A\parr B,C,D, \Gamma}
	\Infer{1}[$\parr$]{A\parr B,C\parr D, \Gamma}
\end{prooftree}
\ \quad and \qquad
\begin{prooftree}
	\Theproof{\pi}{A , B, C,D, \Gamma}
	\Infer{1}[$\parr$]{A, B,C\parr D, \Gamma}
	\Infer{1}[$\parr$]{A\parr B,C\parr D, \Gamma}
\end{prooftree}
\end{center}
are different objects from the point of view of sequent calculus generates the first commutative conversion
we saw above.

A possible solution to this issue is to look for more intrinsic description of proofs, to find a language
that is more \emph{synthetic}; if possible to the point where we have no commutative conversions to perform
anymore.

Introduced at the same time as linear logic, the theory of \emph{proofnets}~\cite{Girard1987,Girard1996}
partially addresses this issue. The basic idea is to describe proofs as graphs rather than trees,
where application of logical rules become local graph construction, thus erasing some inessential sequential
informations. Indeed, the two proofs above would translate into the same proofnet:
\begin{proofnet}
\pnsomenet[R]{\large $\mathcal R_\pi$}{3cm}{1.2cm}
\pnoutfrom{R.-27}{$\Gamma$}[3.6]
\pnoutfrom{R.-50}[D]{$D$}
\pnoutfrom{R.-97}[C]{$C$}
\pnoutfrom{R.-140}[B]{$B$}
\pnoutfrom{R.-155}[A]{$A$}
\pnpar{A,B}{$A\parr B$}
\pnpar{C,D}{$C\parr D$}
\end{proofnet}
(where $\mathcal R_\pi$ is the proofnet translation of the rest of the proof) and the corresponding
commutative conversion disappears.

For the multiplicative without units fragment of linear logic (\mllm), proofnets yield an entirely satisfactory
solution to the problem, and constitute a low-complexity canonical representation of proofs based on local
operations on graphs.

By canonical, we mean here that two proofs are equivalent \modulo the permutations of rules induced by the commutative
conversions
if and only if they have the same proofnet translation.
From a categorical perspective, this means that proofnets constitute a syntactical 
presentation of the free semi-$\ast$-autonomous
category and a solution to the associated word problem~\cite{Heijltjes2014}.

Contrastingly, the linear logic community
has struggled to extend the notion of proofnets to wider fragment: even the question of 
\mll (that is, \mllm plus the multiplicative units) could not find a satisfactory answer.
A recent result~\cite{Heijltjes2014a} helps to understand this situation: proof equivalence of \mll is
actually a \Pspace-complete problem. Hence,
there is no hope for a satisfactory notion of low-complexity proofnet for this fragment%
\footnote{Of course, this applies only to the standard formulation of units: the equivalence problem
for any notion of multiplicative units enjoying less
permutations of rules could potentially still be tractable \via proofnets: see for instance the work 
of S.~Guerrini and A.~Masini~\cite{Guerrini2001} and D.~Hughes~\cite{Hughes2005a}}. 

In this article, we consider the same question, but in the case of \mallm: the multiplicative-additive
without units fragment of linear logic. Indeed, this fragment has so far
also resisted the attempts to build a notion of proofnet that at the same time characterizes proof equivalence
and has basic operations of tractable complexity: we have either canonical nets of exponential size~\cite{Hughes2005}
or tractable nets that are not canonical~\cite{Girard1996}. Therefore, it would have not been too surprising
to have a similar result of completeness for some untractable complexity class. An obvious candidate in that
respect would be \coNP: as we will see, one of these two approaches to proofnets for \mallm is related to Boolean formulas,
which equivalence problem is \coNP-complete.

It turns out in the end that this is not the case: our investigation concludes that the equivalence problem in \mallm
is \Logspace-complete under \aco reductions. But maybe more importantly, we uncover 
in the course of the proof an unexpected connexion
of this theoretical problem with a very practical issue: indeed we show that \mallm proofs are closely related
to binary decision diagrams.%, and vice versa.

	\subsection*{Binary decision diagrams}
	The problem of the representation of Boolean functions is of central importance in circuit design and has
a large range of practical applications. Over the years, binary decision diagrams (\bdd)~\cite{Bryant1986} became the
most widely used data structure to treat this question.

Roughly speaking, a \bdd is a binary tree with nodes labeled by Boolean variables and leaves labeled by
values $\zero$ and $\one$. Such a tree represents a Boolean function in the sense that once an assignment
of the variable is chosen, then following the left or right path at each node according to the value
$\zero$ or $\one$ chosen for its variable eventually leads to a leave, which is the output of the function.

This representation has many advantages which justify its popularity~\cite{Knuth2009}: 
most basic operations (negation and other
logical connectives) on \bdd can be implemented efficiently, in many practical cases the size of the \bdd representing
a Boolean function remains compact (thanks to the possibility to have shared subtrees) and once a variable
ordering is chosen they enjoy a notion of normal form.

In this article, we consider both \bdd and ordered \bdd with no sharing of subtrees and write them
as special kinds of Boolean functions
by introducing an \texttt{IfThenElse} constructor. However, when manipulating them from a complexity point of view
we will keep the binary tree presentation in mind.

	\subsection*{\aco reductions}
	To show that a problem is complete for some complexity class \compclass C, one needs to specify the notion
of reduction functions considered, and of course this needs to be a class of functions supposed to be smaller
than \compclass C itself (indeed \emph{any} problem in \compclass C is complete under \compclass C reductions).

A standard notion of reduction for the class \Logspace is (uniform)
\aco reduction~\cite{Chandra1984}, 
formally defined in terms of 
uniform circuits of fixed depth and unbounded fan-in. We will not be getting in the details
about this complexity class and, as we will consider only graph
transformations, % on graphs of bounded in/out-degree,
we will rely on the following intuitive
principle: if a graph transformation locally replaces each vertex by a bounded number of vertices
and the replacement depends \emph{only} on the vertex considered and eventually its direct neighbors,
then the transformation is in \aco.
Typical examples of such a transformation are certain simple cases of so-called \enquote{gadget} reductions used in complexity theory
to prove hardness results.

	\subsection*{Outline of the paper}
	\Cref{sec_equiv} covers some background material on \mallm and notions of proofnet for this fragment: monomial
proofnets and the associated vocabulary for Boolean formulas and \bdd in \cref{sec_mono} and the notion of slicing 
in \cref{sec_slicing}. 
Then, we introduce in \cref{sec_bool} an intermediary notion of proof representation that will help us to relate
proofs in \mallm and \bindd.
In \cref{sec_bdd}, we prove that proof equivalence in \mallm and equivalence 
of \bindd relate to each other through \aco reductions
and that they are both \Logspace-complete.

\section{Proof equivalence in \mallm}\label{sec_equiv}
	
\newcommand{\Axiom}[1]{\Hypo{}\Infer{1}[\texttt{ax}]{#1}}

\begin{notation}
	The formulas of \mallm are built inductively from atoms which we write $\alpha,\beta,\gamma,\dots$
	their duals $\nalpha,\nbeta,\ngamma,\dots$ and the binary connectives $\parr,\otimes,\labwith x,\plus$
	\emph{(we consider that the $\twith$ connectives carry a label $x$ to simplify some reasonings, but we will
	omit it when it is not relevant)}.
	We write formulas as uppercase letters $A,B,C,\dots$ unless we want to specify they are atoms.
	\emph{Sequents} are sequences of formulas, written as greek uppercase letters $\Gamma,\Delta,\Lambda,\dots$
	such that all
	occurrences of the connective $\twith$ in a sequent carry a different label. The concatenation
	of two sequents $\Gamma$ and $\Delta$ is simply written $\Gamma,\Delta$.
\end{notation}

Let us recall the rules of \mallm\footnote{We consider a $\eta$-expanded version of \mallm,
which simplifies proofs and definitions, but the extension of our results to a version with non-atomic axioms
would be straightforward. Also, we work \modulo the exchange rule.}. We do not include the \cut rule in our study, since in a static situation
(we are not looking at the \cut-elimination procedure of \mallm) it can always be encoded \wloss
using the $\otimes$ rule.
\begin{center}
	\begin{prooftree}
		\Hypo{\alpha,\nalpha}
	\end{prooftree}
	\hfill
\begin{prooftree}
		\Hypo{\Gamma,A,B}
		\Infer{1}[$\rparr$]{\Gamma,A\parr B}
	\end{prooftree}
	\hfill
	\begin{prooftree}
		\Hypo{\Gamma,A}
		\Hypo{\Delta,B}
		\Infer{2}[$\otimes$]{\Gamma,\Delta,A\otimes B}
	\end{prooftree}
	\hfill
	\begin{prooftree}
		\Hypo{\Gamma,A}
		\Infer{1}[$\lplus$]{\Gamma,A\plus B}
	\end{prooftree}
	\hfill
	\begin{prooftree}
		\Hypo{\Gamma,B}
		\Infer{1}[$\rplus$]{\Gamma,A\plus B}
	\end{prooftree}
	\hfill
	\begin{prooftree}
		\Hypo{\Gamma,A}
		\Hypo{\Gamma,B}
		\Infer{2}[$\rlabwith x$]{\Gamma,A\labwith x B}
	\end{prooftree}
\end{center}
\textit{(by convention, we leave the axiom rule implicit to lighten notations.
Also, we will use the notation \begin{prooftree}\Theproof\pi\Gamma\end{prooftree}
for \enquote{the proof $\pi$ of conclusion $\Gamma$}.
)}

\begin{remark}\label{rem_proofs}
	Any time we will look at a \mallm proof from a complexity perspective, we will consider they are represented as
	trees with nodes corresponding to rules, labeled by the connective introduced and the sequent that is 
	the conclusion of the rule. 
%	The cases of \texttt{If}\:$\zero$, \texttt{DontCare}\:$\one,$\dots{} will be useful to obtain
%	\aco reductions in \cref{sec_bool} and \cref{sec_logspace}, since erasing a whole subpart of a graph is not something that is possible
%	in this complexity
%	class.
\end{remark}

%\smallskip
Two \mallm proofs $\pi$ and $\nu$ are said to be \emph{equivalent} (notation $\pi\malleq\nu$) if one can pass
from one to the other \via permutations of rules~\cite{Hughes2015}. We have an associated decision problem.

\newcommand{\mallmeq}{\mallm\problem{equ}\xspace}
\begin{definition}[\mallmeq]\label{def_mallmeq}
	\mallmeq is the decision problem: 
	\begin{center}
	{\it\enquote{Given two \mallm proofs $\pi$ and $\nu$ with the same conclusion, 
	do we have $\pi\malleq\nu$?}}
	\end{center}
\end{definition}

We will not go through all the details about this syntactic way to define proof equivalence in \mallm.
The reason
for this is that we already have an available equivalent characterization in terms of \emph{slicing}~\cite{Hughes2015}
which we review in \cref{sec_slicing}. Instead, let us focus only on the most significant case.

%The proof
\begin{center}
\begin{prooftree}
	\Theproof\pi{\Gamma,A,C}
	\Theproof\mu{\Gamma,B,C}
	\Infer{2}[$\rwith$]{\Gamma,A\with B,C}
	\Theproof\nu{\Delta,D}
	\Infer{2}[$\otimes$]{\Gamma,\Delta,A\with B,C\otimes D}
\end{prooftree}
%is equivalent to the proof
\quad{\large$\malleq$}\qquad
\begin{prooftree}
	\Theproof\pi{\Gamma,A,C}
	\Theproof\nu{\Delta,D}
	\Infer{2}[$\otimes$]{\Gamma,\Delta,A,C\otimes D}
	\Theproof\mu{\Gamma,B,C}
	\Theproof\nu{\Delta,D}
	\Infer{2}[$\otimes$]{\Gamma,\Delta,B,C\otimes D}
	\Infer{2}[$\rwith$]{\Gamma,\Delta,A\with B,C\otimes D}
\end{prooftree}
\end{center}
In the above equivalence, the $\otimes$ rule gets lifted above the $\with$ rule. But doing so, notice that
we created two copies
of $\nu$ instead of one, therefore the size of the prooftree has grown. Iterating on this observation, it is not hard to
build pairs of proofs that are equivalent, but one of which is exponentially bigger than
the other.
This is indeed where the difficulty of proof equivalence in \mallm lies. As a matter of fact, this
permutation of rules \emph{alone} would be enough to build the encoding of the equivalence problem 
of binary decision diagrams presented in \cref{sec_reduce}.

A way to attack proof equivalence in a logic, as we exposed in \cref{sec_twoface}, is to try to setup a notion of
proofnet for this logic. In the following, we will review the main two approaches to this idea in the case of \mallm:
\emph{monomial proofnets} by J.-Y.~Girard~\cite{Girard1996,Laurent2008a} and 
\emph{slicing proofnets} by \Hughes and \VGlab~\cite{Hughes2005,Hughes2015}.
We will then design an intermediate notion of \emph{\bdd slicing} that will be more suited to our needs.

	\subsection{Monomial proofnets}\label{sec_mono}
	The first attempt in the direction of a notion of proofnet for \mallm is due to J.-Y.Girard~\cite{Girard1996},
followed by a version with a full \cut-elimination procedure by O.~Laurent and R.~Maielli~\cite{Laurent2008a}.

While proofnets for multiplicative linear logic without units were introduced along linear 
logic itself~\cite{Girard1987}, extending the notion to the multiplicative-additive without units fragment
proved to be a true challenge, mainly
because of the \emph{superposition} at work in the $\twith$ rule.

Girard's idea was to represent the superposed \enquote{versions} of the proofnet by attributing a Boolean formula
(called a \emph{weight}) to each link, with one Boolean variable for each $\twith$ connective in the conclusion $\Gamma$.
To retrieve the version of the proofnet corresponding to some selection of the left/right branches
of each $\twith$, one then just needs to evaluate the Boolean formulas with the corresponding valuation of their
variables.

This is the occasion to introduce the vocabulary to speak about Boolean formulas.

\begin{definition}[Boolean formula]
	Given a finite set of variables $V=\extset{{x_1}}{x_n}$, a \emph{Boolean formula} over $V$ is inductively defined
	from the elements of $V$; the constants $\zero$ (\emph{\enquote{false}}) and $\unit$(\emph{\enquote{true}});
	the unary symbol $\neg{\cdot}$ (\emph{\enquote{negation}});
	the binary symbols $+$ and $\p$ (\emph{\enquote{sum/or/disjuction}} and 
	\emph{\enquote{product/and/conjunction}} respectively).
%	
%	\TODO{vérifier: égalité modulo commutativité?}
%	Boolean formulas are considered \emph{equal} up to commutativity of $+$ and $\p$,
%	%; absorption and
%	%neutrality of $\zero$ (\ie $\zero\p\phi=\zero$ and $\zero+\phi=\phi$) 
%	but \emph{not} up to other usual boolean equations (distributivity of $\p$ over $+$, for instance).
\end{definition}

%\TODOEXEMPLE

We consider a syntactic
notion of \emph{equality} of Boolean formulas: for instance $\zero\p x\neq\zero$.
The real important notion, that we therefore state separately, is \emph{equivalence}: the fact that if we
replace the variables with actual values, we gets the same output.

\begin{definition}[equivalence]\label{def_booleq}
	A \emph{valuation} $v$ of $V$ is a choice of $\zero$ or $\unit$ for any element of $V$. A valuation induces a
	an \emph{evaluation} function $v(\cdot)$ from Boolean formulas over $V$ to $\{\zero,\unit\}$ in the obvious way.
	Two Boolean formulas $\phi$ and $\psi$ over $V$ are \emph{equivalent} (notation $\phi\booleq\psi$) when
	for any valuation $v$ of $V$, we have $v(\phi)=v(\psi)$.
\end{definition}

\begin{definition}[monomial]
	We write $\neg V=\extset{\neg {x_1}}{\neg x_n}$.
	A \emph{monomial} over $V$ is a Boolean formula
	of the form $y_1\p\,\dots\,\p y_k$ with $\extset{y_1}{y_k}\subseteq V \cup \neg V$.
	%, such that the $y_i$ are
	%parwise distinct and there is no $j$ such that $x_j,\neg x_j\in\extset{y_1}{y_k}$. 
%	By convention
%	the empty monomial is defined to be the Boolean formula $\unit$.
	
	Two monomials $\monof m$ and $\monof m'$ are in \emph{conflict} if $\monof m\p\monof m'\booleq \zero$.
%	If they are not in conflict, they are called \emph{compatible}.
%	A \emph{disjunctive normal form (DNF)} Boolean formula over $V$ is a sum of monomials over~$V$.
\end{definition}

\begin{remark}\label{rem_compat}
	Two monomials are in conflict if and only if there is a variable $x$ such that $x$ appears in one of them and $\neg x$
	appears in the other.
\end{remark}

While monomials are a specific type of Boolean formula, the \emph{binary decision diagrams}
we are about to introduce are not defined directly as Boolean formulas. Of course, they relate to each other in an 
obvious way, but having a specific syntax for binary decision diagrams will prove more convenient to solve the problems we will be facing.

\begin{definition}[\bindd]
	A \emph{binary decision diagram} (\bindd) is defined inductively as:
	\begin{itemize}
		\item The constants $\zero$ and $\one$ are \bindd
%		\item If $\phi$ is a \obdd V, then $\dummy\phi$ is a \bdd.
%		\item If $\phi$, $\psi$ are \bindd then $\itef \zero\phi\psi$ and $\itef \one\phi\psi$ are \bdd.
		\item If $\phi$, $\psi$ are \bindd and $X$ is either a variable, $\zero$ or $\one$, $\itef X\phi\psi$ is a \bindd
		\item If $\phi$ is a \bindd and $X$ is either a variable, $\zero$ or $\one$, $\dcare X\phi$ is a \bindd
	\end{itemize}
	Moreover, suppose we have an ordered set of variables $V=\extset{{x_1}}{x_n}$ with the convention
	that variables are listed in the reverse order: $x_n$ is first, then $x_{n-1}$, \etc
	We define a subclass of
	\bindd which we call \emph{ordered binary decision diagrams over $V$} (\obdd V) by
	restricting to the following inductive cases (we let $V'=\extset{{x_1}}{x_{n-1}}$)
	\begin{itemize}
		\item The constants $\zero$ and $\one$ are \obdd \void
		\item If $\phi$ and $\psi$ are \obdd {V'}, $\itef {x_n}\phi\psi$ is a \obdd V
		\item If $\phi$ is a \obdd{V'}, $\dcare {x_n}\phi$ is a \obdd V
	\end{itemize}
	The notions of valuation and equivalence are extended to \bindd and \obdd{} the obvious way so that
	$\dcare X\phi\,\booleq\,\phi$ and $\itef X\phi\psi\,\booleq\, X\p\phi+\neg X\p\psi$.
\end{definition}

\begin{remark}
	Any time we will look at \bdd and \obdd{} from a complexity perspective, we will consider they are represented as
	labeled trees. The cases of \texttt{If}\:$\zero$, \texttt{DontCare}\:$\one,$\dots{} will be useful to obtain
	\aco reductions in \cref{sec_bool} and \cref{sec_logspace}, since erasing a whole subpart of a graph of which we do
	not know the address in advance
	is not something that is doable in this complexity class. The absence of sharing implied by the
	tree representation is also crucial to get low-complexity reductions.
\end{remark}

\begin{example}
	The Boolean formula $x\p y\p\neg z$ is a monomial, while $x\p y +z$ and $x\p\one$ are not.
	
	The \bdd $\itef {x_2}{(\itef {x_1}\zero\one)}{\one}$ (which is \emph{not} a \obdd{\{{x_1},{x_2}\}} by the way)
	is equivalent to the Boolean formula
	${x_2}\p\neg{x_1} +\neg {x_2}$: both evaluate to $\one$ only for the valuations $\{{x_1}\mapsto \one,{x_2}\mapsto \zero\}$,
	$\{{x_1}\mapsto \zero,{x_2}\mapsto \zero\}$ and $\{{x_1}\mapsto \zero,{x_2}\mapsto \one\}$. There exist an equivalent
	\obdd{\{{x_1},{x_2}\}}: $\itef {x_2}{(\itef {x_1}\zero\one)}{(\dcare {x_1}\one)}$.
\end{example}

%These two notions come with associated equivalence problems.
\newcommand{\bddequ}{\bdd\problem{equ}\xspace}
\newcommand{\obddequ}{\protectobdd\problem{equ}\xspace}
\begin{definition}\label{def_bddequ}
	\bddequ is the following decision problem: 
	\begin{center}
	{\it\enquote{given two \bdd $\phi$ and $\psi$, do we have
	$\phi\booleq\psi$?}}
	\end{center}
	\obddequ is the following decision problem: 
	\begin{center}
	{\it\enquote{given two \obdd V $\phi$ and $\psi$, do we have
	$\phi\booleq\psi$?}}
	\end{center}
\end{definition}

Girard's proofnets are called monomial because the only Boolean formulas that are allowed are monomials.
This is, as of the state of the art, the only known way to have a notion of proofnet that enjoys a satisfying 
correctness criterion.

For our purposes, we do not need to get into the technical details of monomial proofnets. Still, let
us end this section with an example of proofnet from the article by Laurent and Maielli,
where the monomial weight of a link is pictured just above it:
%~\cite{Laurent2008a},
%so that we have 
%at least 
%an idea of how they look anyway:

\begin{center}
	\includegraphics[width=9cm]{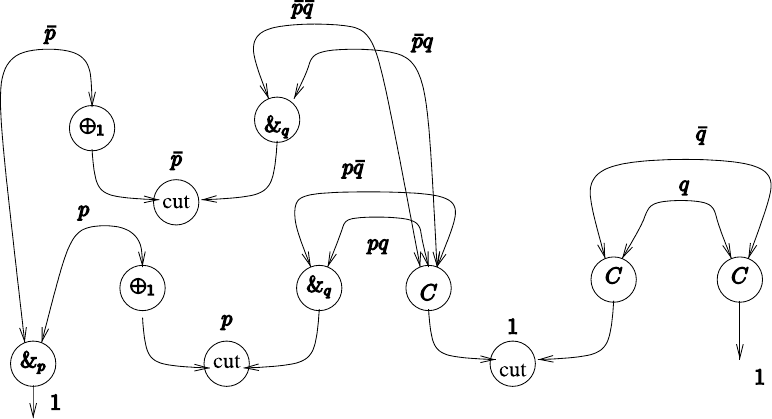}
\end{center}

	\subsection{Slicings and proof equivalence}\label{sec_slicing}
	The idea of slicing dates back to J.-Y.~Girard's original article on proofnets for \mallm~\cite{Girard1996},
and was even present in the original article on linear logic~\cite{Girard1987}.
It amounts to the
natural point of view already evoked in \cref{sec_mono}, seeing the $\twith$ rule as introducing
superposed variants of the proof,
which are eventually to be selected from in the course of \cut-elimination.
If we have two alternative \emph{slices} for each $\twith$ connective of a sequent $\Gamma$ and all
combinations of slices can be selected independently, we readily see that the global number of slices will be
exponential in the number of $\twith$ connectives in $\Gamma$.

This is indeed the major drawback of the representation of proofs as set of slices: the size of objects
representing proofs may grow exponentially in the size of the original proofs. This of course impairs any fine-grained
analysis in terms of complexity.

\begin{definition}[slicing]
	Given a \mallm sequent $\Gamma$, a \emph{linking} of $\Gamma$ is a subset
	$$\extset[\big]{[\alpha_1,\nalpha_1]}{[\alpha_n,\nalpha_n]}$$
	of %$\dual(\Gamma)$,
	the set of (unordered) pairs of occurrences of dual atoms in $\Gamma$.
	
	Then, a \emph{slicing} of $\Gamma$ is a finite set of linkings of $\Gamma$.
	
	To any \mallm proof $\pi$, we associate a slicing $\slicing\pi$ by induction:

	\begin{itemize}
		\item If $\pi=\begin{prooftree}\Hypo{\alpha,\nalpha}\end{prooftree}$ then 
		      $\slicing\pi$ is the set containing only the linking $\big\{[\alpha,\nalpha]\big\}$
		
		\medskip
		\item If $\pi=\,
		      \begin{prooftree}
		      \Theproof{\mu}{\Gamma,A,B}
		      \Infer{1}[$\rparr$]{\Gamma,A\parr B}
		      \end{prooftree}$
		      \ \begin{minipage}[t]{10cm}
		      then $\slicing\pi=\slicing\mu$, where we see atoms 
		      of $A\parr B$ \\as the corresponding atoms of $A$ and $B$
		      \end{minipage}
		
		\medskip
		\item If $\pi=\,
		      \begin{prooftree}
		      \Theproof{\mu}{\Gamma,A}
		      \Theproof{\nu}{\Delta,B}
		      \Infer{2}[$\otimes$]{\Gamma,\Delta,A\otimes B}
		      \end{prooftree}$
		      then 
		      $\slicing\pi=\set{\lambda\cup\lambda'}{\lambda\in \slicing\mu\,,\,\lambda'\in\slicing\nu}$
		
		\medskip
		\item If $\pi=\,
		      \begin{prooftree}
		      \Theproof{\mu}{\Gamma,A}
		      \Infer{1}[$\lplus$]{\Gamma,A\plus B}
		      \end{prooftree}$
		      then 
		      $\slicing\pi=\slicing\mu$,
		and likewise for $\rplus$
		
		\medskip
		\item If $\pi=\,
		      \begin{prooftree}
		      \Theproof{\mu}{\Gamma,A}
		      \Theproof{\nu}{\Gamma,B}
		      \Infer{2}[$\rwith$]{\Gamma,\Delta,A\with B}
		      \end{prooftree}$
		      then 
		      $\slicing\pi=\slicing\mu \cup \slicing\nu$
	\end{itemize}
\end{definition}

%\TODOEXEMPLE

\begin{remark}
	In the $\otimes$ rule, it is clearly seen that the number of slices are multiplied. This is just what is
	needed in order to have a combinatorial explosion: for any $n$, a proof $\pi_n$ of
	$$\underbrace{{\nalpha}\otimes\cdots\otimes{\nalpha}}_{n\text{ times}},
	\underbrace{\extlist{\alpha\with \alpha}{\alpha\with \alpha}}_{n\text{ times}}$$
	obtained by combining with the $\otimes$ rule $n$ copies of the proof
	\begin{prooftree*}
		\Hypo{\nalpha,\alpha}
		\Hypo{\nalpha,\alpha}
		\Infer{2}[$\rwith$]{\nalpha,\alpha\with\alpha}
	\end{prooftree*}
	will be of linear size in $n$, but with a slicing $\slicing n$ containing $2^n$ linkings.
\end{remark}

Slicings (associated to a proof) correspond exactly to the notion of proofnets elaborated by
\Hughes and \VGlab~\cite{Hughes2005}. 
While their study was mainly focused on the problems of finding a correctness criterion and
designing a \cut-elimination procedure for these, it also covers the proof equivalence problem.
The proof that their notion of proofnet characterizes \mallm proof equivalence
can be found in an independent note~\cite{Hughes2015}.

\begin{theorem}[slicing equivalence~\protect{\cite[Theorem~1]{Hughes2015}}]\label{th_equiv}
	Let $\pi$ and $\nu$ be two \mallm proofs.
	We have that $\pi\malleq\nu$ if and only if $\slicing\pi=\slicing\nu$.
\end{theorem}

%We will use this characterization in the reductions between proof equivalence in \mallm and equivalence of
%\BDD. 
Let us also end this section
with a graphical representation of an example of proofnet from the article of Hughes and van Glabbeek,
encoding the proof on the left
with three linkings $\lambda_1, \lambda_2, \lambda_3$:

\medskip
\begin{center}
	\scalebox{0.75}{
	\begin{prooftree}[center=false]
		\Hypo{\lneg P,P}
		\Infer{1}[$\lplus$]{\lneg P\oplus \lneg Q,P}
		\Infer{1}[$\lplus$]{(\lneg P\oplus \lneg Q)\oplus \lneg R,P}
		\Hypo{\lneg Q,Q}
		\Infer{1}[$\rplus$]{\lneg P\oplus \lneg Q,Q}
		\Infer{1}[$\lplus$]{(\lneg P\oplus \lneg Q)\oplus \lneg R,P}
		\Infer{2}[$\with$]{(\lneg P\oplus \lneg Q)\oplus \lneg R,P\with Q}
		\Hypo{\lneg R,R}
		\Infer{1}[$\rplus$]{(\lneg P\oplus \lneg Q)\oplus \lneg R,P}
		\Infer{2}[$\with$]{(\lneg P\oplus \lneg Q)\oplus \lneg R,(P\with Q)\with R}
	\end{prooftree}
	}
	\hfill
	\scalebox{0.8}{
	\includegraphics[width=5cm]{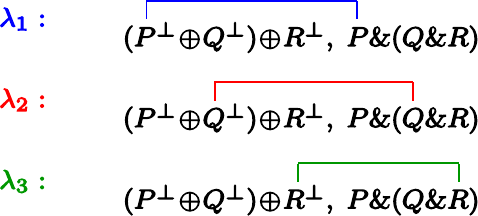}
	}
\end{center}

	\subsection{\bdd slicings}\label{sec_bool}
	We finally introduce an intermediate notion of representation of proofs which will be a central tool
in the next section.
In a sense, it is a synthesis of monomial proofnets and slicings: acknowledging the fact that slicing makes
the size of the representation explode, we rely on \bdd to keep things more compact.

Of course, the canonicity property is lost. But this is exactly the point! Indeed, deciding whether two
\enquote{\bdd slicings} are equivalent is the reformulation of proof equivalence in \mallm we 
rely on
in the reductions between \mallmeq (\cref{def_mallmeq}) and \bddequ (\cref{def_bddequ}).
%to prove that this problem lies in \Logspace.

\begin{definition}[\bdd slicing]\label{def_boolslicing}
	Given a \mallm sequent $\Gamma$, a \emph{\bdd slicing} of $\Gamma$ is a function $\slfont B$ 
	that associates a \bdd to every element $[\gamma,\ngamma]$
	of %$\dual[\Gamma]$,
	the set of (unordered) pairs of occurrences of dual atoms in $\Gamma$.

	We say that two \bdd slicings $M,N$ of the same $\Gamma$ are \emph{equivalent}
	(notation $M\sliceeq N$) if for any pair $[\gamma,\ngamma]$,
	we have $M[\gamma,\ngamma]\booleq N[\gamma,\ngamma]$ in the sense of \cref{def_booleq}.
	
	To any \mallm proof $\pi$, we associate a \bdd slicing $\mslicing\pi$ by induction:
	\begin{itemize}
		\item If $\pi=\begin{prooftree}\Hypo{\alpha,\nalpha}\end{prooftree}$ then 
		      $\mslicing\pi[\alpha,\nalpha]=\unit$.
		
		\medskip
		\item If $\pi=\,
		      \begin{prooftree}
		      \Theproof{\mu}{\Gamma,A,B}
		      \Infer{1}[$\rparr$]{\Gamma,A\parr B}
		      \end{prooftree}$
		      \ \begin{minipage}[t]{10.2cm}
%		      then $\mslicing\pi[\gamma,\ngamma]=\dcare\one{(\mslicing\mu[\gamma,\ngamma])}$ where we see atoms 
		      then $\mslicing\pi[\gamma,\ngamma]=\mslicing\mu[\gamma,\ngamma]$ where we see atoms 
		      of $A\parr B$ as the corresponding atoms of $A$ and $B$
		      \end{minipage}
		
		\medskip
		\item If $\pi=\,
		      \begin{prooftree}
		      \Theproof{\mu}{\Gamma,A}
		      \Theproof{\nu}{\Delta,B}
		      \Infer{2}[$\otimes$]{\Gamma,\Delta,A\otimes B}
		      \end{prooftree}$
		      then 
		      $\mslicing\pi[\gamma,\ngamma]=
		      \begin{cases}
		      \mslicing\mu[\gamma,\ngamma] &\text{ if $\gamma,\ngamma$ are atoms of $\Gamma,A$}\\
		      \mslicing\nu[\gamma,\ngamma] &\text{ if $\gamma,\ngamma$ are atoms of $\Delta,B$}\\
		      \zero &\text{ otherwise\footnotemark}
		      \end{cases}
		      $
%		      \itef\zero{(\mslicing\mu[\gamma,\ngamma])}{(\mslicing\nu[\gamma,\ngamma])}
%		      		&\text{ if $\alpha,\nalpha$ are atoms of $\Gamma,A$}\\
%		      \itef\one{(\mslicing\mu[\gamma,\ngamma])}{(\mslicing\nu[\gamma,\ngamma])}
%		      		&\text{ if $\alpha,\nalpha$ are atoms of $\Delta,B$}\\

		\medskip
		\item If $\pi=\,
		      \begin{prooftree}
		      \Theproof{\mu}{\Gamma,A}
		      \Infer{1}[$\lplus$]{\Gamma,A\plus B}
		      \end{prooftree}$
\begin{minipage}[t]{10cm}
		      then
		      $\mslicing\pi[\gamma,\ngamma]=
		      \begin{cases}
		      \mslicing\mu[\gamma,\ngamma] &\text{ if $\gamma,\ngamma$ are atoms of $\Gamma,A$}\\
		      \zero &\text{ otherwise}
		      \end{cases}
		      $
		      
		      and likewise for $\rplus$.
\end{minipage}

%		      ~\hfill$\mslicing\pi[\gamma,\ngamma]=
%		      \begin{cases}
%		      \dcare\one{(\mslicing\mu[\gamma,\ngamma])} &\text{ if $\alpha,\nalpha$ are atoms of $\Gamma,A$}\\
%		      \zero &\text{ otherwise}
%		      \end{cases}

		\medskip
		\item If $\pi=\,
		      \begin{prooftree}
		      \Theproof{\mu}{\Gamma,A}
		      \Theproof{\nu}{\Gamma,B}
		      \Infer{2}[$\rlabwith x$]{\Gamma,\Delta,A\labwith x B}
		      \end{prooftree}$
%		      \ \begin{minipage}[t]{10cm}
		      then
		      
		      \hfill$\mslicing\pi[\gamma,\ngamma]=
		      \begin{cases}
		      \itef x{\mslicing\mu[\gamma,\ngamma]}\zero &\text{ if $\gamma$ or $\ngamma$ is an atom of $A$}\\
		      \itef x\zero{\mslicing\nu[\gamma,\ngamma]} &\text{ if $\gamma$ or $\ngamma$ is an atom of $B$}\\
		      \itef x{\mslicing\mu[\gamma,\ngamma]}{\mslicing\nu[\gamma,\ngamma]} &\text{ otherwise}
		      \end{cases}
		      $
%		      \end{minipage}
	\end{itemize}
\end{definition}
\footnotetext{Remember we consider \emph{occurrences} of atoms, and as the $\otimes$ rule splits the
context into two independent parts that and no axiom rule can cross this splitting.}

\newcommand{\tparr}{\!\parr\!}
%\begin{remark}
%	The choice of \bdd constructors in the above definition might look a bit contrived at first. Indeed, we
%	could have made obvious simplifying choices, as for instance having simply
%	$\mslicing\mu[\gamma,\ngamma]$ instead of $\dcare\one{(\mslicing\mu[\gamma,\ngamma])}$ in the case of
%	the $\tparr$ rule. But then, we might have been struggling a little to show that the translation is
%	doable in \aco, which would have complicated our completeness proof.
%\end{remark}

\begin{remark}
	The \bdd we obtain this way are actually of a specific type: they are usually called \emph{read-once}
	\bdd: from the root to any leave, one never
	crosses two \texttt{IfThenElse} nodes asking for the value of 
	the same variable.
\end{remark}

\begin{example}
	The weight of the pairs $[\alpha,\nalpha]$ and $[\delta,\ndelta]$ in the \bdd slicing of the proof
	\begin{center}
	$\pi=\ $
	\begin{prooftree}
	\Hypo{\alpha,\nalpha}
	\Infer{1}[$\lplus$]{\alpha\oplus\beta,\nalpha}
	\Hypo{\beta,\nbeta}
	\Infer{1}[$\rplus$]{\alpha\oplus\beta,\nbeta}
	\Infer{2}[$\rlabwith x$]{\alpha\oplus\beta,\nalpha\labwith x\nbeta}
	\Hypo{\delta,\ndelta}
	\Infer{2}[$\otimes$]{\alpha\oplus\beta,(\nalpha\labwith x\nbeta)\otimes\delta,\ndelta}
	\end{prooftree}
	\end{center}
	are $\mslicing\pi[\alpha,\nalpha]=\itef x{\one}{\zero}$ and 
	$\mslicing\pi[\delta,\ndelta]=\one$.
\end{example}

It is not hard to see that proof equivalence matches the equivalence of \bdd slicings by relating
them to slicings from the previous section.

\begin{theorem}[\bdd slicing equivalence]
	Let $\pi$ and $\nu$ be two \mallm proofs.
	We have that $\pi\malleq\nu$ if and only if $\mslicing\pi\sliceeq\mslicing\nu$.
\end{theorem}

\begin{proof} We show in fact that $\slicing\pi=\slicing\nu$ if and only if $\bslicing\pi\booleq\bslicing\nu$,
	with \cref{th_equiv} in mind.
%	Consider the function $f$ from \bdd slicings to slicings:
	\newcommand{\mB}{\mathcal B}
	
	To a \bdd slicing $\mB$, we can associate a linking $v(\mB)$ for each valuation $v$ of the variables
	occurring in $\mB$ by setting
	$v(\mB)=\set{[\alpha,\nalpha]}{v(\mathcal B[\alpha,\nalpha])=\one}$ and then a slicing
	$f(\mB)=\set{v(\mB)}{v \text{ valuation}}$. By definition, it is clear that if $\mB$ and $\mB'$
	involve the same variables and $\mB\sliceeq\mB'$ then $f(\mB)=f(\mB')$.
	
	Conversely, suppose $\bslicing\pi\not\booleq\bslicing\nu$, so that there is a $v$ such that
	$v(\bslicing\pi)\neq v(\bslicing\nu)$. To conclude that $f(\bslicing\pi)\neq f(\bslicing\nu)$, we must show 
	that there is no other $v'$ such that $v'(\bslicing\nu)=v(\bslicing\pi)$.
	
	To do this, we can extend the notion of valuation to proofs: if $v$ is a valuation of
	the labels $x$ of the $\tlabwith x$ in $\pi$, $v(\pi)$ is defined by keeping only the left or right
	branch of $\tlabwith x$ according to the value of $x$.
	Now we can consider the set $v^{\tt ax}(\pi)$ of axiom rules in $v(\pi)$
	% and see that it corresponds exactly to a 
%	linking of $\slicing\pi$.
	and we can show by induction that $v^{\tt ax}(\pi)$ must contain at least one
	pair with one atom which is a subformula of the side of each $\tlabwith x$ that has been kept. Therefore for any
	$\pi$ and $\nu$ with the same conclusion, if $v\neq v'$ we have $v^{\tt ax}(\nu)\neq v'^{\tt ax}(\pi)$
	no matter what.
	Then we can remark that $v^{\tt ax}(\pi)$ is just another name for $v(\bslicing\pi)$
	so that in the end, there cannot be $v\neq v'$ such that 
	$v'(\bslicing\nu)=v(\bslicing\pi)$.
	
	Finally, an easy induction shows that $f(\bslicing\pi)=\slicing\pi$ and therefore we are done.
\end{proof}

%\begin{lemma}[\ite]
%	For any \mallm proof $\pi$ and any pair $[\gamma,\ngamma]$,
%	$\bslicing \pi[\gamma,\ngamma]$ is a \bdd.
%\end{lemma}

Also, a \bdd equivalent to the \bdd associated to a pair can be computed in \aco.

\begin{lemma}[computing \bdd slicings]\label{lem_bsl}
	For any \mallm proof $\pi$ and any pair $[\gamma,\ngamma]$,
	we can compute in \aco a \bdd $\phi$ such that $\phi\booleq\bslicing \pi[\gamma,\ngamma]$.
\end{lemma}

\begin{proof} As we see proofs as labeled trees (\cref{rem_proofs}), we will only locally replace
the rules of the proof the following way to obtain the corresponding \bdd $\phi$:
\begin{itemize}
	\item Replace axiom rules \begin{prooftree}\Hypo{\gamma,\ngamma}\end{prooftree} by $\one$ and other axiom
	rules by $\zero$
	\item Replace all $\tparr$ and $\oplus$ rules by $\dcare\one{\cdot}$ nodes
	\item In the $\otimes$ case, test which side the atoms $\gamma,\ngamma$ are attributed to 
	\emph{(this can be done locally by looking at the conclusions of the premise of the rule)}
	and replace it by
	a $\itef \zero{\cdot}{\cdot}$ or a $\itef \one{\cdot}{\cdot}$ node accordingly
	\item Replace $\tlabwith x$ rules by a $\itef x{\cdot}{\cdot}$ nodes
\end{itemize}
We can see by induction that the resulting \bdd is equivalent to $\bslicing\pi[\gamma,\ngamma]$.
All these operations can be performed by looking only at the rule under treatment (and its immediate neighbors
in the case of $\otimes$) and always replaces one rule by exactly one node. Therefore it is in \aco.
\end{proof}

%Hence, proof equivalence in \mallm reduces in \aco to equivalence of \ite.
\begin{corollary}[reduction]
	\mallmeq reduces to \bddequ in \aco.
\end{corollary}

In the next section we focus on the equivalence of \bdd and \obdd{}, proving first that the case of \obdd{}
can be reduced to proof
equivalence in \mallm.
%
%it to be in the class \Logspace.
Then, we will show the problem of equivalence of \BDD to be in \Logspace, and that of \obdd{} to be
\Logspace-hard,
thus characterizing the intrinsic complexity of proof equivalence in \mallm as \Logspace-complete.

Note that this contrasts with the classical result that equivalence of general Boolean formulas is \coNP-complete.
It turns out indeed that the classes of \BDD we consider enjoy a number of properties that allow to solve 
equivalence
more easily.

%the only thing that is really left to do is to find a way to build proofs that have boolean slicing
%expressing any boolean formula, and make sure that the equivalence of these proofs correspond exactly
%the the equivalence of the boolean formula they express. This is the objective of the next section.

\section{Equivalence of \bdd}\label{sec_bdd}
	\subsection{Equivalence of \protectobdd reduces to proof equivalence in \mallm}\label{sec_reduce}

We now show that the converse of \cref{lem_bsl} holds for \obdd{}.

%so that we finally have a two way \aco reduction
%between proof equivalence in \mallm and equivalence of \bdd.

%proof leading to \coNP-completeness. Although we will get to the technical
%details shortly, let us spare a moment in explaining the basic idea.

To do this, we rely on a formula $\boolf$ 
%(for \enquote{boolean})
 which will serve as the placeholder of a
\obdd{V} $\phi$ we want to encode; and a context $\Gamma$ which contains one $\tlabwith x$ connective for each 
variable $x$ in $V$, organized in a way that allows for an inductive specification of \obdd{}.
%%, organized in a way that allows for the inductive
%\enquote{programming} of boolean formulas.

Given an \obdd{} $\phi$, we wish to obtain a proof $\pdnf \phi$ of $\boolf,\Gamma$ such that
dual pairs with one element in $\boolf$ will receive either the value $\phi$ or a value equivalent to $\neg\phi$
in the \bdd slicing of $\pdnf \phi$;
and on the other hand, the other dual pairs of $\Gamma$ will receive equivalent values whatever the \obdd{}
we encode is. This will lead to the fact that two such encoding proofs are equivalent if and only if
the \obdd{} they encode are equivalent.

%More precisely, we will encode an \ite boolean formulas, and for this we start off by encoding monomials.
%Of course, we need to check at each step that the encoding is computable in polynomial time%
%\footnote{It is even actually computable in logarithmic space, if ones looks carefuly.}.

\begin{notation}
	We fix atomic formulas $\extlist{\alpha_1}{\alpha_n}$\dots{} and $\beta$ and write 
	$\boolf\eqdef \beta\plus\beta$. In what follows, we will use $\lbeta$ and $\rbeta$ to refer
	respectively to the left and right copies of $\beta$ in $\boolf$; %and when the case arises, we will use
	and likewise $\lalpha_i$ and $\ralpha_i$ for copies of $\alpha_i$ in $\alpha_i\with\alpha_i$.
	
	We write respectively $\pleft$ and $\pright$ the proofs
	\begin{center}
		\begin{prooftree}
			\Hypo{\beta,\nbeta}
			\Infer{1}[$\lplus$]{\boolf,\nbeta}
		\end{prooftree}
		\qquad
		\begin{prooftree}
			\Hypo{\beta,\nbeta}
			\Infer{1}[$\rplus$]{\boolf,\nbeta}
		\end{prooftree}
	\end{center}
	and for any $n$, we write $\pwith n$ the proof \ 
	\begin{prooftree*}
		\Hypo{\nalpha_n,\alpha_n}
		\Hypo{\nalpha_n,\alpha_n}
		\Infer{2}[$\rwith$]{\nalpha_n,\alpha_n\with\alpha_n}
	\end{prooftree*}
%	
%	We will also be using the notation
%	\begin{prooftree}\Theproof\pi\beta\end{prooftree}
%	for \enquote{the proof $\pi$ of conclusion $\beta$}.
\end{notation}

\begin{definition}[encoding a \protectobdd]
	Let $\phi$ be an \obdd{} over the variables $V=\extset{x_1}{x_n}$. We define the sequent
	$$
	\seqv n\eqdef\ 
	\big(\cdots(\nbeta\otimes \nalpha_1)\otimes \nalpha_2)\otimes\cdots\big)
		\otimes\nalpha_n \:,\:
	\alpha_1\labwith {x_1}\alpha_1 \:,\:
	\alpha_2\labwith {x_2}\alpha_2 \:,\:
	\dots \:,\:
	\alpha_{n-1}\labwith {x_{n-1}}\alpha_{n-1}
	$$
	with $\seqv 0=\nbeta$
	and set $\seqvd n\eqdef \seqv n,\alpha_n\labwith{x_n}\alpha_n$
	with also $\seqvd 0=\nbeta$.

\newcommand{\LEllipsis}[3]{
	\Infer[rulemargin=0pt,rulestyle=none]{1}{\hspace{#1} \ddots}
	\Hypo{}
	\Infer[hsep=#2,rulestyle=none]{2}{#3}
}

	We define the proof $\pdnf \phi$ of conclusion $\boolf,\seqvd {n}$ by induction on $n$
	
	\begin{itemize}
		\item The base cases are $\zero$ and $\one$, encoded respectively as $\pleft$ and $\pright$.
%			\begin{center}
%				$\prln {\mathit i}n=\quad$
%				\begin{prooftree}
%					\Theproof{\prln {\mathit i}{}}{\boolf,\nbeta}
%					\Theproof{\pwith 1}{\nalpha_1,\alpha_1\with\alpha_1}
%					\Infer{2}[$\otimes$]{\boolf,\seqvd 1}
%					\Theproof{\pwith 2}{\nalpha_2,\alpha_2\with\alpha_2}
%					\Infer{2}[$\otimes$]{\boolf,\seqvd 2}
%					\LEllipsis{20pt}{0pt}{\boolf,\seqvd {n-1}}
%					\Theproof{\pwith n}{\nalpha_n,\alpha_n\with\alpha_n}
%					\Infer{2}[$\otimes$]{\boolf,\seqvd n}
%				\end{prooftree}
%			\end{center}
		\item Otherwise, if $\phi=\itef {x_n}\psi\zeta$, with both $\psi$ and $\zeta$ 
		being \obdd{\extset{x_1}{x_{n-1}}},
		we have 
		$\pdnf \psi$ and $\pdnf \zeta$ defined by induction, and then
\begin{center}
$\pdnf \phi=\ $
			\begin{prooftree}
				\Theproof{\pdnf \psi}{\boolf,\seqvd {n-1}}
				\Hypo{\nalpha_n,\alpha_n}
				\Infer{2}[$\otimes$]{\boolf,\seqv {n},\alpha_n}
				\Theproof{\pdnf \zeta}{\boolf,\seqvd {n-1}}
				\Hypo{\nalpha_n,\alpha_n}
				\Infer{2}[$\otimes$]{\boolf,\seqv {n},\alpha_n}
				\Infer{2}[$\rlabwith n$]{\boolf,\seqvd n}
			\end{prooftree}
\end{center}
		\item The last case is $\phi=\dcare {x_n}\psi$, with $\psi$ being a \obdd{\extset{x_1}{x_{n-1}}}
		so that we have 
		$\pdnf \psi$ defined by induction, and then
\begin{center}
$\pdnf \phi=\ $
			\begin{prooftree}
				\Theproof{\pdnf \psi}{\boolf,\seqvd {n-1}}
				\Theproof{\pwith n}{\nalpha_n,\alpha_n\labwith n\alpha_n}
				\Infer{2}[$\otimes$]{\boolf,\seqvd {n}}
			\end{prooftree}
\end{center}
\end{itemize}
\end{definition}

We still need to state in what sense $\pdnf \phi$ is an encoding of $\phi$: we turn the %vague 
statement about the value of atoms of $\boolf$ we made in the beginning of this section into a precise
property.

\begin{lemma}[associated \bdd slicing]
	Writing $\slfont B$ the \bdd slicing of $\pdnf \phi$, we have
		$$\slfont B[\lbeta,\nbeta]=\phi
		\qquad\quad
		\slfont B[\rbeta,\nbeta]\booleq{\neg{\phi}}
		\qquad\quad
		\slfont B[\nalpha_i,\lalpha_i]\booleq{x_i}
		\qquad\quad
		\slfont B[\nalpha_i,\ralpha_i]\booleq{\neg x_i}$$
\end{lemma}

\begin{proof}
%	This a routine inspection of induction cases, let us review the case where $\monof m=x_n\p\monof m'$.
%
\newcommand{\bstop}{\bslicing {\texttt{s}}}
	By a routine inspection of induction cases. Let us only review $\phi=\itef {x_n}\psi\zeta$.
%	We write $\slfont B'$ and $\slfont B$ the respective boolean slicings of $\pdnf {}$ and $\pdnf m$;
%	and $\bstop$ that of $\pstop n$. 
%	First note that $\bslicing\[\lbeta,\nbeta]=\zero$ while $\slfont B[\rbeta,\nbeta]=\unit$.
	Let us write $\bslicing \phi$, $\bslicing \psi$ and $\bslicing \zeta$ the \bdd slicings respectively
	associated to the proofs $\pdnf \phi$, $\pdnf \psi$ and $\pdnf \zeta$.
	
	By induction we have that $\bslicing\psi[\lbeta,\nbeta]= \psi$ and 
	$\bslicing\zeta[\lbeta,\nbeta]= \zeta$, therefore by definition of the \bdd slicing associated
	to a $\twith$ rule, we have that 
	$\bslicing\phi[\lbeta,\nbeta]= \itef{x_n}\psi\zeta=\phi$. The case of $\rbeta$ is similar, but for its use
	of \cref{lem_neg} from the next section.
%	and 
%	$\slfont B[\rbeta,\nbeta]\booleq x_n\p\neg{\monof m}'+\neg x_n\p\unit\booleq\neg{\monof m}$.

	As for the other pairs of occurences of dual atoms in the conclusion, let us have a look at the case
	of $[\nalpha_i,\lalpha_i]$:
%	the other cases being similar(with $\slfont B[\beta,\nbeta]\booleq\unit$). 
	if by induction $\bslicing\psi[\nalpha_i,\lalpha_i]\booleq x_i$ and 
	$\bslicing\zeta[\nalpha_i,\lalpha_i]\booleq x_i$,
	then by definition $\bslicing\phi[\nalpha_i,\lalpha_i]\booleq \itef{x_n}{x_i}{x_i}\booleq x_i$.
	The case of $\ralpha_i$ is similar.
%	$\slfont B'[\nalpha_i,\ralpha_i]\booleq x_n\p \neg x_i+\neg x_n\p\neg x_i\booleq \neg x_i$.
\end{proof}

\begin{corollary}[equivalence]
	If $\phi$ and $\psi$ are two \obdd{\extset{x_1}{x_n}}, then
	$\pdnf  \phi \malleq \pdnf  \psi$ if and only if $\phi \booleq \psi$.
\end{corollary}

\begin{lemma}[computing the encoding]%[logarithmic space]
	Given a \obdd{\extset{x_1}{x_n}} $\phi$, $\pdnf \phi$ can be computed in \aco.
\end{lemma}

\begin{proof} As in the proof of \cref{lem_bsl}, we show that the inductive definition can in fact be
seen as a local graph transformation introducing nodes of bounded size:
\begin{itemize}
	\item Replace each $\zero$ node by 
	\begin{prooftree}\Theproof{\pleft}{\boolf,\nbeta}\Ellipsis{}{}\end{prooftree} 
	and $\one$ node by 
	\begin{prooftree}\Theproof{\pright}{\boolf,\nbeta}\Ellipsis{}{}\end{prooftree}.
	
	\smallskip
	\item Replace each $\dcare {x_i}{\cdot}$ by \ 
			\begin{prooftree}
				\Hypo{\vdots}
				\Theproof{\pwith {x_i}}{\nalpha_i,\alpha_i\labwith {x_i}\alpha_i}
				\Infer{2}[$\otimes$]{\boolf,\seqvd {i}}
				\Ellipsis{}{}
			\end{prooftree}
	
	\smallskip
	\item Replace each $\itef {x_i}{\cdot}{\cdot}$ by \ 
			\begin{prooftree}
				\Hypo{\vdots}
				\Hypo{\nalpha_i,\alpha_i}
				\Infer{2}[$\otimes$]{\boolf,\seqv {i},\alpha_i}
				\Hypo{\vdots}
				\Hypo{\nalpha_i,\alpha_i}
				\Infer{2}[$\otimes$]{\boolf,\seqv {i},\alpha_i}
				\Infer{2}[$\rlabwith {x_i}$]{\boolf,\seqvd {i}}
				\Ellipsis{}{}
			\end{prooftree}
\end{itemize}
For any of these replacements, we see that the choice of the case to apply and the label of the resulting 
block (of bounded size)
of rules replacing a node depends
only on the label of the node we are replacing, therefore the transformation is in \aco.
\end{proof}

\begin{corollary}[reduction] \obddequ reduces to \mallmeq in \aco.
\end{corollary}

%To sum up, we have proven so far the following bidirectional reduction theorem:

%\begin{theorem}[bidirectional reduction]
%	Proof equivalence in \mallm and equivalence of \bdd reduce to each other in \aco.
%\end{theorem}

To sum up, we have so far the following chain of \aco reductions:

\begin{center}%\footnotesize
	\obddequ\quad\!\!\!$\rightarrow$\quad%
	\mallmeq\!\!\!\quad$\rightarrow$\quad \bddequ
\end{center}

	\subsection{\Logspace-completeness}\label{sec_logspace}
	We prove in this section that all these equivalence problems are \Logspace-complete. We begin by listing
a few useful properties of \bdd that will allow to design a \Logspace decision procedure for their equivalence.
Then, we prove the \Logspace-hardness by reducing to \obddequ a \Logspace-complete problem on line graph orderings.

The starting point is the good behavior of \bdd with respect to negation. In the following lemma, we consider
the negation of a \bdd which is not strictly speaking a \bdd itself:
we think of it as the equivalent Boolean formula, the point
being precisely to show that this Boolean formula can be easily expressed as a \bdd.

\begin{lemma}[negation]\label{lem_neg}
	If $\phi$ and $\psi$ are \bdd and $X$ is either $\zero$, $\one$ or a variable, we have
	$$\neg{\itef X\phi\psi}\,\booleq\,\itef X{\neg\phi}{\neg\psi}$$
	$$\neg{\dcare X\phi}\,\booleq\,\dcare X{\neg\phi}$$
\end{lemma}

\begin{proof}
	First we can transform our expression by
	$$\neg{\itef X\phi\psi}\,\booleq\, \neg{X\p\phi+\neg X\p\psi}\,\booleq\,(\neg X+\neg\phi)\p(X+\neg\psi)
	\,\booleq\, X\p\neg\psi+\neg X\p\neg\phi+\psi\p\phi$$
	then we can apply the so-called \enquote{consensus rule} of Boolean formulas
	$$X\p\neg\psi+\neg X\p\neg\phi+\psi\p\phi\,\booleq\,
	 X\p\neg\phi+\neg X\p\neg\psi\,\booleq\, \itef X{\neg \phi}{\neg\psi}$$
	The case of \texttt{DontCare} is obvious.
\end{proof}

\begin{corollary}
	If $\phi$ is a \ite, then there is a \bdd $\ineg \phi$ such that
	$\neg\phi\booleq\ineg\phi$.
	
	Moreover $\ineg \phi$ can be computed in logarithmic space.
\end{corollary}

\begin{proof}
	An induction on the previous lemma shows that we can obtain the negation of a \bdd simply by flipping
	the $\zero$ nodes to $\one$ nodes and conversely. Hence the transformation is even in \aco.
\end{proof}

Then, we show that a \bdd can be seen as a sum of monomials through a \Logspace transformation.

\begin{lemma}[\ite as sums of monomials]
	If $\phi$ is a \bdd, then there is a formula $\idnf\phi$ which is a sum of monomials
	and is such that $\phi\booleq\idnf\phi$.
	
	Moreover $\idnf \phi$ can be computed in logarithmic space.
\end{lemma}

\begin{proof}
	For each $\one$ node in $\phi$, go down to the root of $\phi$ and output one by one the variables of
	any $\itef x{\cdot}{\cdot}$ encountered: this produces the monomial associated to this $\one$ node. Then
	$\idnf \phi$ is the sum of all the monomials obtained this way and is clearly equivalent to $\phi$.
	The procedure is in \Logspace because we only need to remember which $\one$-leave we are treating and where we are
	in the tree (when going down) at any point.
\end{proof}

Putting all this together, we finally obtain a space-efficient decision procedure. Note however that it is totally
sub-optimal in terms of time: to keep with the logarithmic space bound, we have to recompute a lot of things
rather than store them.

\begin{corollary}[\bddequ is in \Logspace]
	There is a logarithmic space algorithm that, given two \bdd $\phi$ and $\psi$, decides
	wether they are equivalent.
\end{corollary}

\begin{proof}
	The \bdd $\phi$ and $\psi$ are equivalent if and only if 
	$\phi\Leftrightarrow\psi=(\neg\phi+\psi)\p(\neg\psi+\phi)\booleq 1$ that is to say (by passing to the
	negation)
	$(\neg\psi\p\phi)+(\neg\phi\p\psi)\booleq \zero$, which holds if and only if
	both $\neg\psi\p\phi\booleq\zero$ and $\neg\phi\p\psi\booleq \zero$.

	But then, considering the first one (the other being similar) we can rewrite it in logarithmic space
	using the two above lemmas as $\idnf{(\ineg\psi)}\p\idnf{\phi}\booleq\zero$.
	This holds if and only if for all pairs $(\monof m,\monof m')$ of one monomial in 
	$\idnf{(\ineg\psi)}$ and one monomial
	in $\idnf{\phi}$, $\monof m$ and $\monof m'$ are in conflict; which can be checked in logarithmic space
	using \cref{rem_compat}.
\end{proof}

Let us now introduce an extremely simple, yet \Logspace-complete problem~\cite{Etessami1997}, which will ease the
\Logspace-hardness part of our proof.

\newcommand{\ord}{\problem{ORD}\xspace}
\begin{definition}[order between vertices]
\emph{Order between vertices} (\ord) is the following decision problem:
%\shrinkspace
	\begin{center}
	{\it\enquote{%
	Given a directed graph $G=(V,E)$ that is a line%
	\footnote{We use the standard definition of graph as a pair $(V,E)$ of sets of vertices and edges (oriented
	couples of vertices $x\rightarrow y$).
	A graph is a \emph{line} if it is connected and all the vertices have in-degree and out-degree $1$, except the
	\emph{begin} vertex which has in-degree $0$ and out-degree $1$ and the \emph{exit} vertex
	which has in-degree $1$ and out-degree $0$. A line induces a total order on vertices through its transitive
	closure}
	and two vertices $f,s\in V$\\
	do we have 
	$f<s$ in the total order induced by $G$?}}
	\end{center}
\end{definition}

\begin{lemma}
	\ord reduces to \obddequ in \aco.
\end{lemma}

\begin{proof}
	Again we are going to build a local graph transformation that is in \aco.
	
	First, we assume \wloss that the begin $b$ and the exit $e$ vertices of $G$ are different from $f$ and $s$. We write
	$f^+$ and $s^+$ the vertices immediately after $f$ and $s$ in $G$.
	
	Then, we perform a first transformation by replacing the graph with three copies of itself
	(this can be done by locally scanning the graph and create labeled copies of the vertices and edges).
	We write $x_i$
	to refer to the copy of the vertex $x$ in the graph $i$.
	The second transformation is a rewiring of the graph as follows: erase the edges going out of the $f_i$
	and $s_i$ and replace them as pictured in the two first subgraphs:
%	$$f_1 \rightarrow f_2^+\qquad f_2 \rightarrow f_3^+\qquad f_3 \rightarrow f_1^+$$
%	$$s_1 \rightarrow s_2^+\qquad s_2 \rightarrow s_1^+\qquad s_3 \rightarrow s_3^+$$
\newcommand{\compress}{\vspace{-2pt}}
\compress
\begin{center}
	\begin{tikzpicture}
\tikzstyle{localstyle}=[draw=black,thick,->]
\tikzset{every node/.style={inner sep=1pt,font=\small}}
	\matrix(m)[row sep=15pt,column sep=10pt]{%
  \node(f1){$f_1$}; & \node(f2){$f_2$}; & \node(f3){$f_3$}; 
&&&\node(s1){$s_1$}; & \node(s2){$s_2$}; & \node(s3){$s_3$}; \\
  \node[inner sep=-1pt](f1+){$f_1^+$}; & \node[inner sep=-1pt](f2+){$f_2^+$}; & \node[inner sep=-1pt](f3+){$f_3^+$};
&&&\node[inner sep=-1pt](s1+){$s_1^+$}; & \node[inner sep=-1pt](s2+){$s_2^+$}; & \node(s3+){$s_3^+$}; \\
};
\path [localstyle] (f1) to (f2+);
\path [localstyle] (f2) to (f3+);
\path [localstyle] (f3) to (f1+.north east);
\path [localstyle] (s1.-50) to (s2+);
\path [localstyle] (s2.-140) to (s1+);
\path [localstyle] (s3) to (s3+);
	\end{tikzpicture}

\qquad\qquad\qquad
	\begin{tikzpicture}
\tikzstyle{localstyle}=[draw=black,thick,->]
\tikzset{every node/.style={inner sep=0pt,font=\small,text depth=1pt}}
	\matrix(m)[row sep=8pt,column sep=6pt]{%
&&\node(x){$x$};\\
&&&\node[text height=5pt,text depth=2pt](y){$y$}; \\
\node[inner sep=-1pt](b1){$b_1$};&& \node[inner sep=-1pt](b2){$b_2$};&& \node[text height=7pt](b3){$b_3$}; \\};
\path [localstyle] (x.-50) to (y);
\path [localstyle] (x.-130) to (b1);
\path [localstyle] (y) to (b2);
\path [localstyle] (y) to (b3);
	\end{tikzpicture}

\end{center}
\compress
Let us call $G_r$ the rewired graph and $G_n$ the non-rewired graph. To each of them we add two binary nodes
$x$ and $y$ connected to the begin vertices $b_i$ as pictured in the third graph above.
%
%\compress
%\begin{center}
%
%\end{center}
%%
%\compress
Then we can produce two corresponding \obdd{} $\phi_r$ and $\phi_n$ by replacing the exit vertices $e_1$, $e_2$, $e_3$ by $\one$, $\zero$, $\zero$ 
respectively, $x$ and $y$ by a $\itef {x}{(\dcare y\cdot)}{(\itef y\cdot\cdot)}$ block of nodes; and
any other $v_i$ vertex by a $\dcare v{\cdot}$. 
It is then easy to see that if $f<s$ in the order induced by $G$ if and only if $\phi_r$ and $\phi_n$ are equivalent.

Let us illustrate graphically what happens in the case where $f<s$: we draw the resulting \obdd{} as a labeled
graph with the convention that a node labeled with $z$ with out-degree $1$ is a $\dcare z\cdot$ and a node labeled
with $z$ with out-degree $2$ is a $\itef z\cdot\cdot$ node with the upper edge corresponding to the \texttt{Then}
branch and the lower edge corresponding to the \texttt{Else} branch.
\begin{center}
\begin{tikzpicture}
	\tikzstyle{localstyle}=[draw=black,thick,->]
\tikzset{every node/.style={inner sep=0.5pt,font=\small,text depth=0pt,text height=4pt}}
	\matrix(m)[row sep=7pt,column sep=15pt]{%
&&\node(b1){$b$}; & \node(d1){$\cdots$}; & \node(f1){$f$};
&&\node(f1+){$f^+$}; & \node(dd1){$\cdots$}; & \node(s1){$s$};
&&\node(s1+){$s^+$}; & \node(ddd1){$\cdots$}; & \node(e1){$\one$}; \\
&\node(y1){$y$};\\
\node(x){$x$};&&\node(b2){$b$}; & \node(d2){$\cdots$}; & \node(f2){$f$};
&&\node(f2+){$f^+$}; & \node(dd2){$\cdots$}; & \node(s2){$s$};
&&\node(s2+){$s^+$}; & \node(ddd2){$\cdots$}; & \node(e2){$\zero$}; \\
&\node(y2){$y$};\\
&&\node(b3){$b$}; & \node(d3){$\cdots$}; & \node(f3){$f$};
&&\node(f3+){$f^+$}; & \node(dd3){$\cdots$}; & \node(s3){$s$};
&&\node(s3+){$s^+$}; & \node(ddd3){$\cdots$}; & \node(e3){$\zero$}; \\
};
%% fleches ... 1ere colonne
\path [localstyle] (b1) to (d1);
\path [localstyle] (d1) to (f1);
\path [localstyle] (b2) to (d2);
\path [localstyle] (d2) to (f2);
\path [localstyle] (b3) to (d3);
\path [localstyle] (d3) to (f3);
%fleches f f+
\path [localstyle] (f1) to (f2+.west);
\path [localstyle] (f2) to (f3+.west);
\path [localstyle] (f3.north east) to (f1+.south west);
%% fleches ... 2e colonne
\path [localstyle] (f1+) to (dd1);
\path [localstyle] (dd1) to (s1);
\path [localstyle] (f2+) to (dd2);
\path [localstyle] (dd2) to (s2);
\path [localstyle] (f3+) to (dd3);
\path [localstyle] (dd3) to (s3);
%fleches s s+
\path [localstyle] (s1) to (s2+.west);
\path [localstyle] (s2) to (s1+.west);
\path [localstyle] (s3) to (s3+.west);
%% fleches ... 3e colonne
\path [localstyle] (s1+) to (ddd1);
\path [localstyle] (ddd1) to (e1);
\path [localstyle] (s2+) to (ddd2);
\path [localstyle] (ddd2) to (e2);
\path [localstyle] (s3+) to (ddd3);
\path [localstyle] (ddd3) to (e3);
% fleches xy
\path [localstyle] (x) to (y1);
\path [localstyle] (x) to (y2);
\path [localstyle] (y1) to (b1);
\path [localstyle] (y1) to (b2);
\path [localstyle] (y2) to (b3);

\end{tikzpicture}

\end{center}
\vspace{-0.5cm}
\end{proof}

\begin{remark}
	The above construction relies on the fact that there are non-commuting permutations on the set of three
	elements: in a sense we are just attributing two non-commuting $\sigma$ and $\tau$ to $f$ and $s$ and make
	sure that the order in which they intervene affects the equivalence class of the resulting \obdd{}. An approach
	quite similar in spirit with the idea of \emph{permutation branching program}~\cite{Barrington1989}.
\end{remark}

We can now extend our chain of reductions with the two new elements from this section
\begin{center}%\footnotesize
	\ord\ {\footnotesize(\Logspace-hard)}\quad$\rightarrow$\quad 
	 \obddequ\quad $\rightarrow$\quad 
	 \mallmeq\quad $\rightarrow$\quad \bddequ\ {\footnotesize($\in$\,\Logspace)}
	 %
%	 \footnotesize(\Logspace-hard)\hspace{6.4cm}($\in$\,\Logspace)\hspace{0.6cm}~
\end{center}
so in the end we get our main result:

\begin{theorem}[\Logspace-completeness]
	The decision problems
	\obddequ, \mallmeq and \bddequ are
	\Logspace-complete under \aco reductions.
\end{theorem}

\section{Conclusion}
	In this work, we characterized precisely the complexity of proof equivalence in \mallm as
\Logspace-complete, contrasting greatly with the situation for the \mll fragment which has a \Pspace-complete
equivalence problem. We did so by establishing a correspondence between \mallm proofs and specific classes of \bdd.

This path we took for proving our result is interesting in itself since the established correspondence
allows a transfer of ideas in both directions. In particular, any progress in the theoretical problem of finding a
correct notion of proofnet for \mallm would yield potential applications to \bdd, a notion of widespread
practical use.
An idea to explore might be the notion of \emph{conflict net} defined by D.~Hughes in an unpublished 
note~\cite{Hughes2008}.
Roughly speaking, the principle is to consider a presentation of proofnets with the
information of the links that cannot be present at the same time, rather than giving an explicit
formula to compute their presence, as it is the case with monomial proofnets or the \bdd slicings we introduced.

On the other hand, since many optimization problems regarding \bdd are known to be \NP-complete,
a finer view at the encoding of \cref{sec_reduce} in addition to basic constraints about what we
expect from a notion of proofnet should lead to an impossibility result, even though the equivalence
problem for \mallm is only \Logspace-complete.

%As of their current development, conflict net do not characterize fully proof equivalence in \mallm, but one
%may wonder whether there is a way to fix this.

\bibliographystyle{plain}
\bibliography{biblio}

\end{document}